\documentclass[11pt]{article}

\usepackage{algorithm}
\usepackage[noend]{algpseudocode}
\usepackage{amsmath, amsfonts, graphicx}
\usepackage{enumerate}

\newtheorem{theorem}{Theorem}
\newtheorem{lemma}[theorem]{Lemma}

\newcommand{\lbn}{\mathrm{lbn}}
\newcommand{\rbn}{\mathrm{rbn}}

\newcommand{\vertex}{\mathrm{vertex}}
\newcommand{\slice}{\mathrm{slice}}

\newcommand{\algtable}{\mathrm{table}}
\newcommand{\algmerge}{\mathrm{merge}}
\newcommand{\algcreate}{\mathrm{create}}
\newcommand{\algextend}{\mathrm{extend}}
\newcommand{\algcontract}{\mathrm{contract}}
\newcommand{\algadjust}{\mathrm{adjust}}

\newenvironment{proof}{\noindent{\bf Proof:}}{
\hspace*{\fill} $\Box$ \vskip \belowdisplayskip}

\title{The Densest $k$ Subgraph Problem in $b$-Outerplanar Graphs} 

\author{Sean Gonzales and Theresa Migler}

\begin{document}

\maketitle

\begin{abstract}
We give an exact $O(nk^2)$ algorithm for finding the densest $k$ subgraph in outerplanar graphs. We extend this to an exact $O(nk^2 8^b)$ algorithm for finding the densest $k$ subgraph in $b$-outerplanar graphs. Finally, we hypothesize that Baker's PTAS technique will not work for the densest $k$ subgraph problem in planar graphs.
\end{abstract}

\section{Introduction}

The {\em density} of a graph is defined to be the ratio of edges to
vertices in the graph. More precisely, if an undirected graph
$G= (V,E)$ has $|V|=n$ vertices and $|E|=m$ edges, the density of $G$
is ${m\over n}$. 

In a network that represents academic collaboration, authors within
the densest component of the graph tend to be the most
prolific~\cite{newman2004}. Dense components in a web graph might
correspond to sets of web sites dealing with related
topics~\cite{gkr1998} or link farms~\cite{gkt2005}. A dense subgraph
of the dolphin social network might identify the most popular
dolphins. Finding dense subgraphs aids in finding price value
motifs~\cite{djdlt2009}. Dense subgraphs can identify communities in
social networks~\cite{cs2010}. In the field of visualization, finding
dense subgraphs allows for better graph
compression~\cite{bc2008}. Dense subgraphs assist in finding stories
and events in micro-blogging streams such as
Twitter~\cite{asks2012}. Dense subgraphs can be used to discover
regulatory motifs in genomic DNA~\cite{fnbb2005}, and to find
correlated genes~\cite{llpbszs2005}.  It is therefore interesting to
find the dense components of graphs, or {\em dense subgraphs}. The
density of a subgraph induced by a vertex set $S\subseteq V$ is
$d(S) = {{|E(S)|}\over {|S|}}$. Goldberg gave an $O(nm \log n\log ({{n^2}\over{m}}))$ algorithm to find the
subgraph of maximum density using network flow
techniques \cite{goldberg1984}. 

Given an undirected graph $G=(V,E)$ and an integer $k$, the {\em
  densest $k$ subgraph problem} is defined as follows: find a subgraph
$H\subseteq G$ such that $|V(H)|=k$ and the density of $H$ is
maximized. This problem can be shown to be NP-Hard by a reduction from
Clique \cite{fkp1999}. Since (if $P\neq NP$) there is no polynomial
time algorithm for this problem, we search for an approximation
algorithm. The best approximation algorithm that one could hope to
develop is a {\em PTAS} or {\em polynomial time approximation scheme}. A
PTAS is an algorithm that takes as an input an instance of an
optimization (maximization) problem and a parameter $\epsilon$ and in
polynomial time produces a solution that is within $1-\epsilon$ of the
optimal solution. Unfortunately there is no PTAS (unless NP$\subseteq
\cap_{\epsilon >0}$BPTIME$(2^{n^\epsilon})$) for the densest
$k$-subgraph problem. This is shown by a reduction from the Minimum
Distance of Code Problem \cite{khot2006}. The best approximation known
is a recent $O(n^{1/4 + \epsilon})$-approximation algorithm that runs
in time $n^{O(1/\epsilon)}$ \cite{bccfv2010}, which is an improvement
on the long standing $O(n^{1/3})$ approximation~\cite{fkp1999}.

Often, when there is no PTAS possible for a problem in a general
graph, one will find a PTAS for that same problem restricted to {\em
  planar graphs}, graphs that can be drawn in the plane in such a way
that no two edges cross each other. For example, the maximum
independent set problem, which is known to be NP-Complete in general
graphs, has a PTAS in planar graphs \cite{baker1994}. There are many
PTASes for classic problems such as: travelling
salesperson~\cite{klein2008}, Steiner tree~\cite{bkm2009}, and Steiner
forest~\cite{bhm2009}, when the domain is restricted to be
planar. This leads us to consider planar variants of the densest
subgraph problems. 

While the complexity of the unconnected densest $k$ subgraph problem
on planar graphs is unknown \cite{cp1984}, the connected planar
densest $k$ subgraph problem is NP-Complete by a reduction from the
connected vertex cover problem on planar graphs with maximum degree 4
\cite{kb1991}.

In what follows we begin with an exact polynomial time algorithm for
the densest $k$ subgraph problem in outerplanar graphs inspired by
Baker~\cite{baker1994} in Section~\ref{sec:outerplanar}. We extend this
algorithm to $b$-outerplanar graphs in
Section~\ref{sec:b_outerplanar}. Finally we hypothesize that a Polynomial Time
Approximation Scheme (PTAS) for the densest-k subgraph problem in
planar graphs cannot be achieved using Baker's technique in Section~\ref{sec:PTAS}.

\section{An Algorithm to Find the Densest $k$ Subgraph Problem in
  Outerplanar Graphs}\label{sec:outerplanar}

We modify the dynamic programming technique of Baker~\cite{baker1994} for
outerplanar graphs. Suppose we are given an {\em outerplanar graph}, a
planar graph with an embedding such that all vertices are on the outer
face. We will define a rooted labelled tree, $T$, that will correspond to
the given outerplanar graph. This tree construction was given by Baker~\cite{baker1994}. We repeat it here for completeness.   

Suppose we are given an outerplanar graph, $G$. If $G$ has
any bridges (an edge whose removal disconnects the graph), $(x,y)$,
add an additional edge $(x,y)$. Then the bridge may be treated as a
face. We will call edges {\em exterior} if they lie on the outer face
and {\em interior} otherwise.

Now we build the tree, $T$: First we suppose that $G$ has no
cutpoints (vertices whose removal disconnects the graph). For a description of the tree construction on an outerplanar graph with cutpoints, see~\cite{baker1994}. We place a tree vertex
in each interior face (call these {\em interior tree vertices}) and on
each exterior edge (call these {\em exterior tree vertices}). We add a tree edge between each interior tree vertex and the interior
tree vertices of adjacent faces. We also add a tree edge between
each interior tree vertex and any exterior tree vertices with edges
adjacent to the face. For a very simple example, see Figure~\ref{fig:outerplanar_tree
  construction}.

\begin{figure}
  \includegraphics[scale =
  .25]{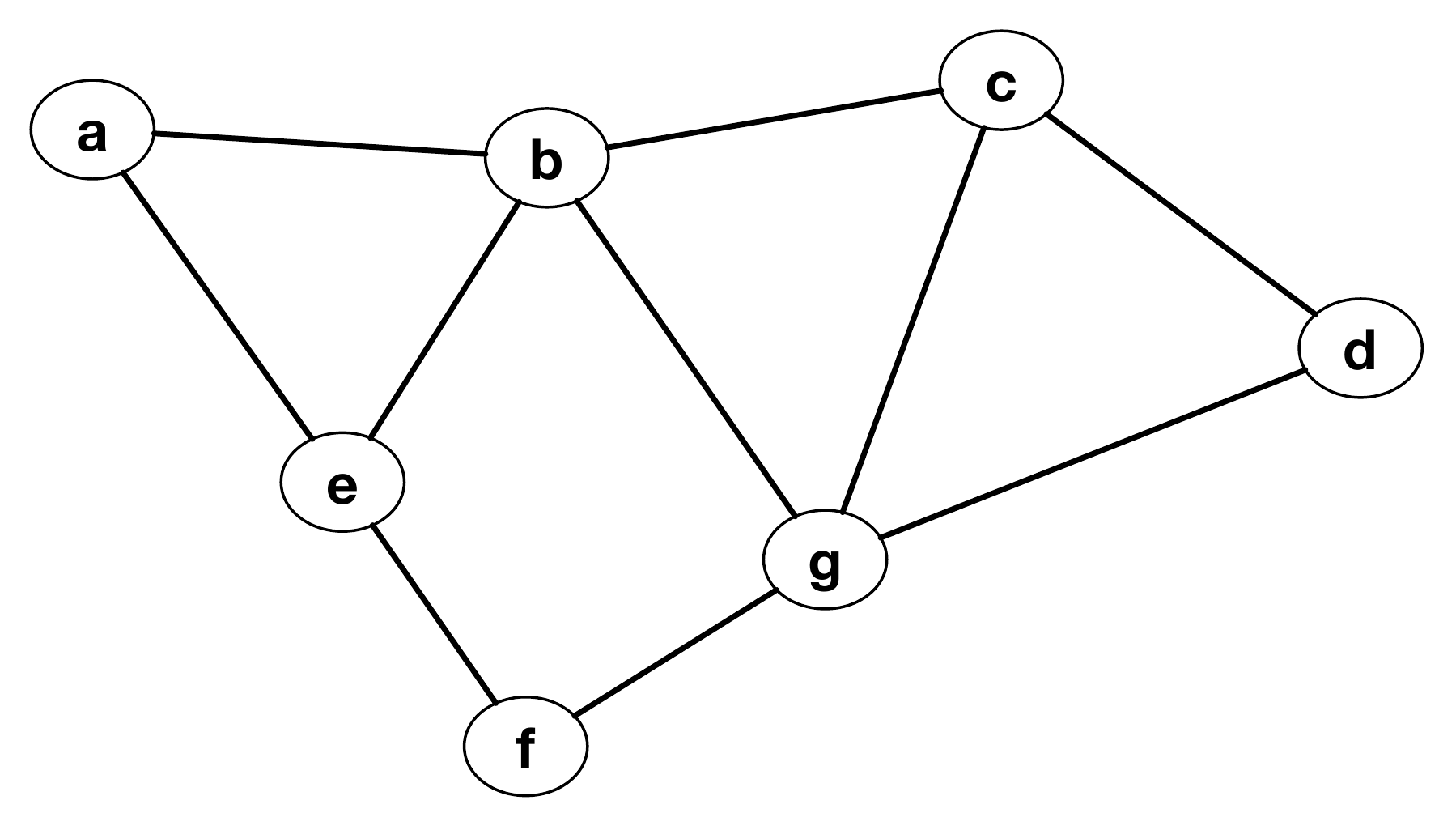}
  \includegraphics[scale
  =.25]{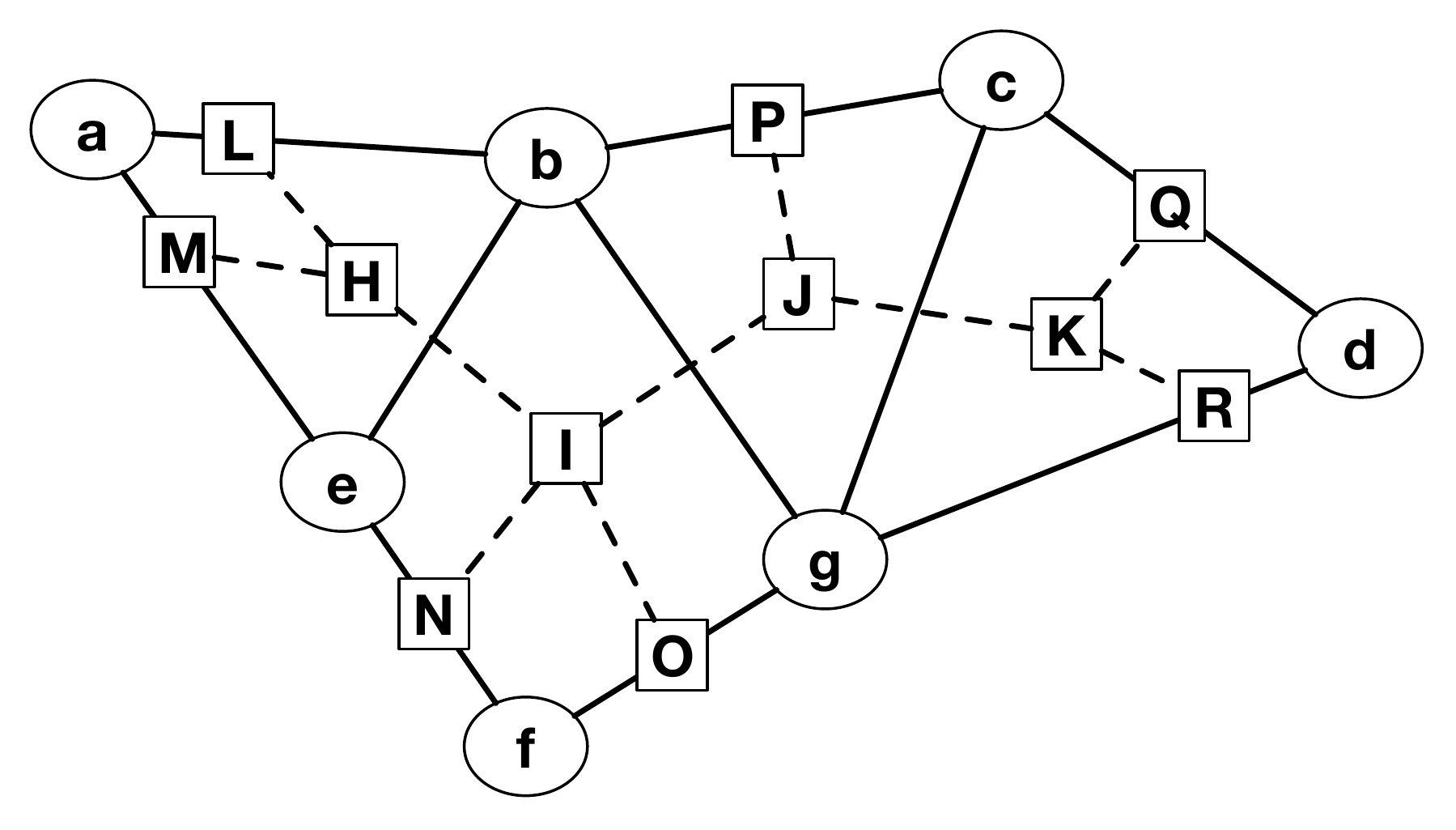}
\includegraphics[scale = .25]{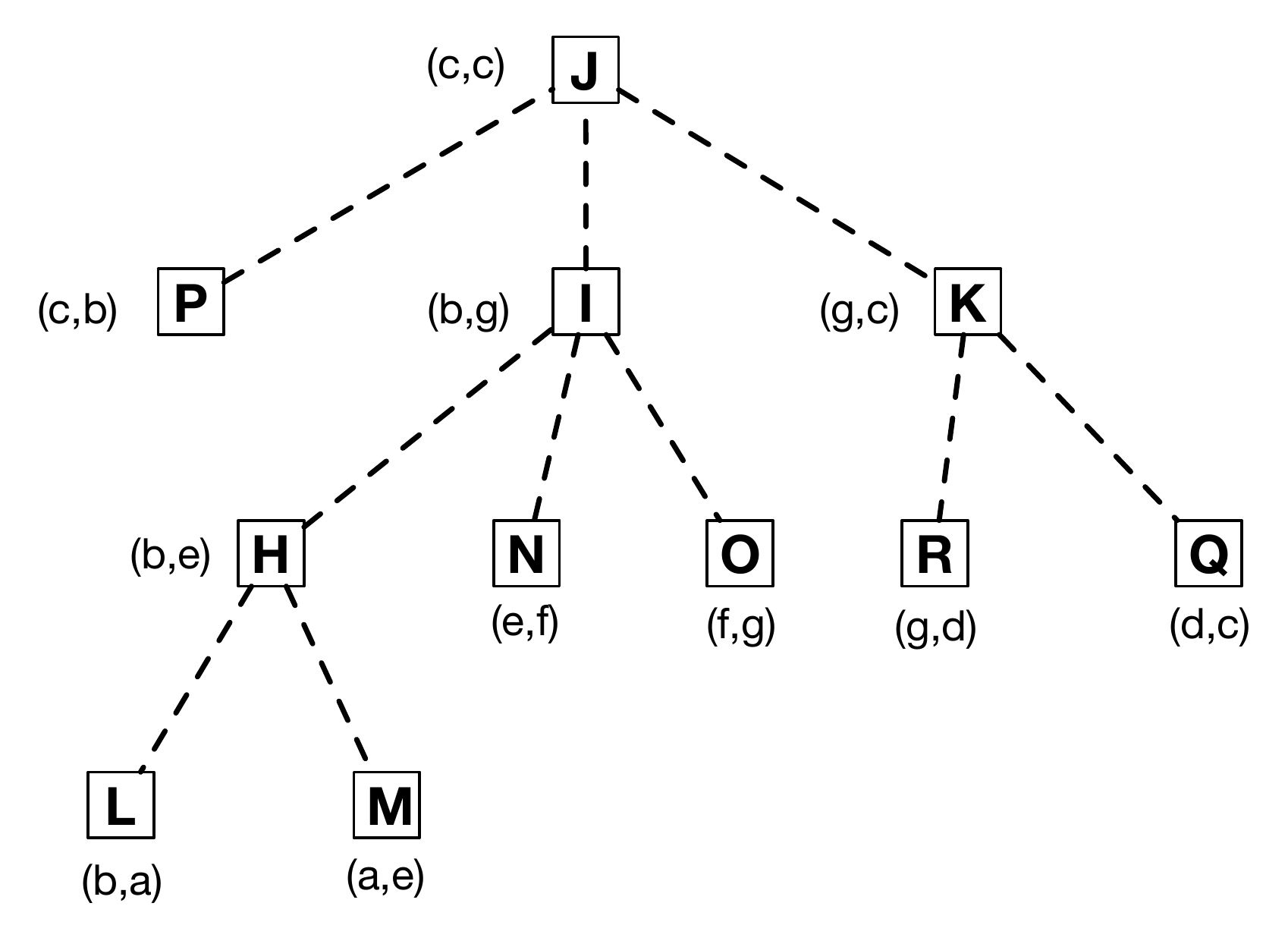}

  \caption{The left figure is an example outerplanar graph, $G$. The
    middle figure shows the construction of $T$ (with square vertices
    and dashed edges) from $G$. The right figure shows the labeling of
  $T$ after making the choice to have vertex $J$ as the root with $P$
  as $J$'s leftmost child.}
\label{fig:outerplanar_tree construction}
\end{figure}

We may choose any interior tree vertex to be the root of $T$. We may
also choose which child of this root will be the leftmost child. These
two choices determine the ordering of all remaining vertices. We label the vertices of $T$ recursively. Label each exterior tree
vertex in the tree with the exterior edge that it is on. Label each
interior tree vertex with the first and last vertices of its
children's labels.

After constructing Baker's tree (as described above), we give our original dynamic program. We fill in a table for each vertex, $v$, in $T$ (with label
$(x,y)$). The table will hold the maximum number of edges for a
subgraph on $k$ vertices ($k$ ranging
from 0 to the size of the subgraph for the subtree rooted at $v$)
depending on whether or not $x$ and $y$ are in the set. For example,
consider the leaf vertex, $L$, in $T$ representing edge $(b,a)$ in
$G$:

$Table(L) = T_{(b,a)} = $
\begin{tabular}{|l|l|l|l|l|}
\hline
$b$ & $a$ & $k=0$ & $k=1$ & $k=2$ \\ \hline
0   & 0   &   0                      &   $\emptyset$    &   $\emptyset$    \\ \hline
0   & 1   &   $\emptyset$    &    0                     &    $\emptyset$   \\ \hline
1   & 0   &   $\emptyset$    &    0                     &    $\emptyset$   \\ \hline
1   & 1   &   $\emptyset$    &    $\emptyset$   &   1 \\ \hline
\end{tabular}

This table is undefined for many entries, for example, the entry where
$b=0$, $a=1$, and $k=2$ is undefined because there is no way to have a
subgraph on 2 vertices when there is only one vertex to select
from. For the entry where $b=1$, $a=1$, and $k=2$, we obtain a value
of $1$ because there is an edge between $b$ and $a$. The table will be
identical for each leaf in $T$.

Now we will fill out a table for a tree vertex with exactly two
children. This table is calculated by \textit{merging} the tables for its two children, as described below. Consider $H$ with label $(b,e)$:

$Table(H) = T_{(b,e)} =$
\begin{tabular}{|l|l|l|l|l|l|}
\hline
$b$ & $e$ & $k=0$       & $k=1$       & $k=2$       & $k=3$       \\ \hline
0   & 0   & 0           & 0           & $\emptyset$ & $\emptyset$ \\ \hline
0   & 1   & $\emptyset$ & 0           & 1           & $\emptyset$ \\ \hline
1   & 0   & $\emptyset$ & 0           & 1           & $\emptyset$ \\ \hline
1   & 1   & $\emptyset$ & $\emptyset$ & 1           & 3           \\ \hline
\end{tabular}

We now give pseudocode for how to {\em merge} two sibling tables, $T_{(x,y)}$ and
$T_{(y,z)}$ (creating $T_{(x,z)}$). For the table of a vertex with label $\mathcal{L}$,  $\mathcal{K}$ will be the minimum of $k$ (the input $k$ for the densest $k$ subgraph problem) and the number of vertices in the subtree represented by label $\mathcal{L}$.

\begin{algorithm}[H]
\caption{Merge($(T_{(x,y)},
  T_{(y,z)},\mathcal{K})$)}\label{alg:merge_outerplanar}
\begin{algorithmic}
  \For {Each $(b_x,b_z)$ in $\{(0,0),(0,1),(1,0),(1,1)\}$}

\State {$values = [ ]$}

\For {Each $k$ from 0 to $\mathcal{K}$}

\For {Each $b_y$ in $\{0,1\}$}

\For {Each $k_x$ from $b_x+b_y$ to $k$}

\State {$k_z = k - k_x +b_y$}

\If {$x == z$ AND $(b_x$ AND $b_y)$}

\State {$k_z ++$}

\EndIf

\State {$value = T_{(x,y)}(b_x,b_y,k_x) + T_{(y,z)}(b_y,b_z,k_z)$}

\If {$x\neq z$ AND $(b_x$ AND $b_z)$ AND $(x,z)$ is an edge}

\State {$value ++$}

\EndIf

\If{$value \neq \emptyset$}

\State {Add $value$ to $values$}

\EndIf

\EndFor

\EndFor

\EndFor

\If {$values$ is not empty}

\State {$T_{(x,z)}(b_x,b_z,k) =max(values)$}

\Else

\State {$T_{(x,z)}(b_x,b_z,k) = \emptyset$}

\EndIf
  \EndFor
\end{algorithmic}

\end{algorithm}

To find the solution for the densest $k$ subgraph problem, we look to the table for the root, in the case of our example, the table for $J$. We take the maximum value in the table for the column corresponding to $k$.

$Table(J) = T_{(c,c)} =$
\begin{tabular}{|l|l|l|l|l|l|l|l|l|l|}
\hline
$c$ & $c$ & $k=0$       & $k=1$       & $k=2$       & $k=3$       & $k=4$       & $k=5$       & $k=6$       & $k=7$       \\ \hline
0   & 0   & 0           & 0           & 1           & 3           & 4           & 6           & 7           & $\emptyset$ \\ \hline
0   & 1   & $\emptyset$ & $\emptyset$ & $\emptyset$ & $\emptyset$ & $\emptyset$ & $\emptyset$ & $\emptyset$ & $\emptyset$ \\ \hline
1   & 0   & $\emptyset$ & $\emptyset$ & $\emptyset$ & $\emptyset$ & $\emptyset$ & $\emptyset$ & $\emptyset$ & $\emptyset$ \\ \hline
1   & 1   & $\emptyset$ & 0           & 1           & 3           & 4           & 6           & 8           & 10          \\ \hline
\end{tabular}

\vspace{.5cm}
To find all intermediate tables, please see Appendix~\ref{app:outerplanar_tables}.

The following two lemmas will give the proof of correctness of the dynamic programming algorithm. 

\begin{lemma}
\label{outer_lemma1}
For each leaf vertex, $X$ (with label $(a,b)$), in the tree, $Table(X) = T_{(a,b)}$ correctly gives the maximum number of edges in the subgraph corresponding to $X$ with the constrained number of vertices (each $k$ value).
\end{lemma}

\begin{proof}
Each leaf vertex, $X$, in the tree corresponds to a single edge, $(a,b)$, in the original graph. There are only three possible values of $k$ in this case, zero, one, or two. 

If there are zero vertices included, the maximum number of edges is zero. However if we try to include either vertex $a$ or $b$ or both this would be absurd (because we are including zero vertices), so these entries should be undefined.

If one vertex is to be included, again it would be absurd to try to include zero or two vertices, so these entries are undefined. The only feasible possibilities are including only $a$ for zero edges or only $b$ for zero edges.

If two vertices are to be included, it would be absurd to include zero or only one vertex, so these entries are undefined. The only feasible possibility is to include both $a$ and $b$ and since there is an edge between $a$ and $b$, we get one edge.
\end{proof}

\begin{lemma}
For each vertex, $V$, in the tree $Table(V)$ correctly gives the maximum number of edges in the subgraph corresponding to $V$ with the constrained number of vertices (each $k$ value).
\end{lemma}

\begin{proof}
If $V$ is a leaf vertex, then Lemma \ref{outer_lemma1} gives us that $Table(V)$ is correct. So, we may suppose that $V$ is not a leaf vertex and that the tables for the children of $V$ are correct. Let $X_1, X_2, \dots, X_{m-1}$ denote the children of $V$, and let the child $X_i$ have label $(x_i, x_{i+1})$. This means that $V$ has label $(x_1,x_m)$. The table for $V$ is then calculated by merging the tables for each $X_i$; that is, we first merge $Table(X_1) = T_{(x_1,x_2)}$ and $Table(X_2) = T_{(x_2,x_3)}$ to get a table $T_{(x_1,x_3)}$, then we merge $T_{(x_1,x_3)}$ with $Table(X_3) = T_{(x_3,x_4)}$ to get $T_{(x_1,x_4)}$ and so on, ultimately obtaining the table $T_{(x_1,x_m)}$. We claim that a call to \textit{merge} on the tables for two consecutive tree vertices $X$ and $Y$ with labels $(x,y)$ and $(y,z)$, respectively, results in a correct table for the union of the subgraphs corresponding to $X$ and $Y$, as well as the edge $(x,z)$ (if it is an edge). If the claim is true, then the result of merging the tables of each child $X_i$ must result in the correct table for $V$.

The table $T_{(x,z)}$ should have a row for each pair of values $b_x$ and $b_z$, and each entry in the row corresponds to a particular value of $k$. So, fix $b_x$, $b_z$, and $k$. The entry $T_{(x,z)}(b_x,b_z,k)$ should contain the maximum number of edges possible over all subgraphs of $k$ vertices that may contain $x$ and $z$ (depending on the values of $b_x$ and $b_z$), where this subgraph is contained in the union described above. The \textit{merge} procedure checks for the cases when we include $y$ or do not include $y$ in the subgraph by setting the bit $b_y$. So, fix $b_y$. The procedure then checks the egde count in every possible subgraph by varying the values of $k_x$ and $k_z$, which are the $k$ values we use to lookup edge counts in the tables $T_{(x,y)}$ and $T_{(y,z)}$, respectively. The idea is that we can check all possible subgraphs of size $k$ containing the included vertices by checking subgraphs corresponding to $X$ and $Y$ such that the number of vertices in both subgraphs sums to $k$. In order to ensure that the values of $k_x$ and $k_z$ do in fact result in a correctly sized subgraph union, the procedure first sets $k_x$ to be some value between $b_x + b_y$ and $k$ (the lower bound is due to the fact that $k_x$ is a vertex count, and $b_x$ and $b_y$ effectively count whether $x$ and $y$ are included). The procedure then sets $k_z$ to be $k - k_x + b+y$; adding in $b_y$ avoids double counting the inclusion of $y$ (if it is included). However, if $x = z$ and $b_x = b_z = 1$ (i.e. $x$ and $z$ are both included in the subgraph), we have double counted the vertex $x$, hence the procedure increments the value of $k_z$ to allow another vertex to be included in the union. The procedure then saves the value $T_{(x,y)}(b_x,b_y,k_x) + T_{(y,z)}(b_y,b_z,k_z)$ and increments it if $(x,z)$ is an edge and both vertices are included (this is because each individual table does not account for the edge). Once all such values are computed, the procedure sets $T_{(x,z)}(b_x,b_z,k)$ to be the maximum value (if all values were undefined, the table entry is set to be undefined as well). Since each value corresponds to an edge count of the corresponding subgraph and all subgraphs were checked, the table entry is thus the maximum number of edges. Therefore, the \textit{merge} procedure results in a correct table for the corresponding subgraph, whence the table for $V$ must also be correct.
\end{proof}

Note that the running time for this dynamic program on outerplanar graphs is, $O(nk^2)$: There is a table for each tree vertex and the number of tree vertices is the number of edges in the graph (plus 1 for the root which doesn't correspond to an edge). Also, since our graph is planar the number of edges is linear in the number of vertices. It takes $O(k^2)$ time to fill in each table.

\section{An Algorithm to Find the Densest $k$ Subgraph Problem in
  $b$-Outerplanar Graphs}\label{sec:b_outerplanar}

We now lay out the dynamic programming solution for the densest $k$ subgraph problem in $b$-outerplanar graphs. We define trees, labels, and slices as given by Baker \cite{baker1994}. We assume that the given graph is connected. We define {\em levels}. Level 1 vertices form the outer face of the $b$-outerplanar graph. Level 2 vertices form the outer face if all level 1 vertices are removed, and so on. Additionally, we assume that the level $i$ vertices contained in a level $i-1$ face form a connected subgraph called a \textit{level $i$ component}. If this is not the case, we may add fake edges that will simply be ignored when we calculate table values. Throughout this section, we will use the 3-outerplanar graph in Figure \ref{fig:kout_ex} as a running example.

Furthermore, we need to construct a triangulation of our graph. We will use the triangulation to construct trees in Section \ref{koutTrees} and slices in Section \ref{koutSlices}. Any triangulation will do, and one can be constructed in linear time by scanning the vertices in levels $i$ and $i+1$ in parallel for each $i = 1, 2, \dots, k-1$. A triangulation for the graph in Figure \ref{fig:kout_ex} is given in Figure \ref{fig:kout_ex_tri}.

\begin{figure}[h]
\begin{minipage}{.5\textwidth}
\centering
\includegraphics[width=\textwidth]{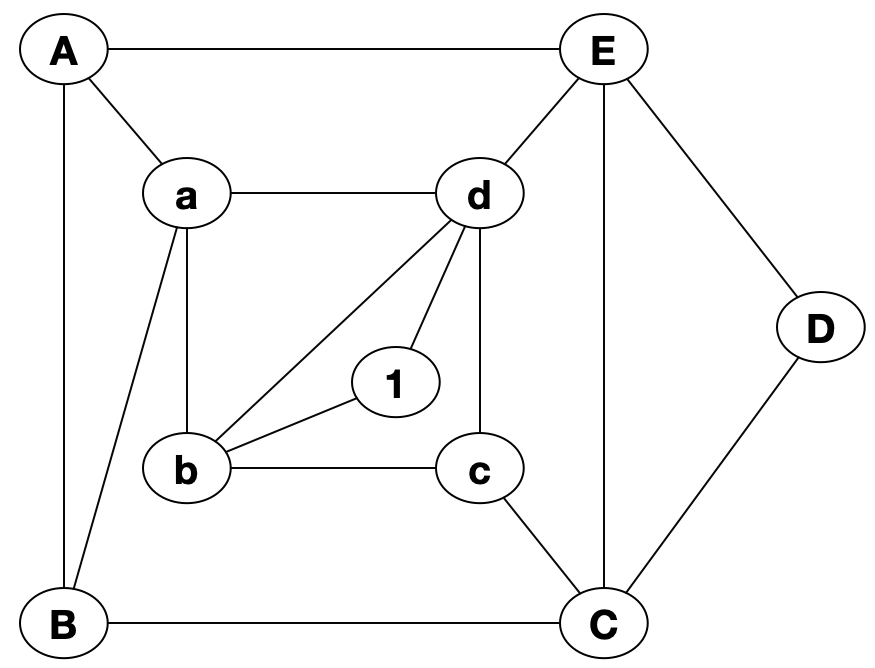}
\caption{figure}{3-outerplanar example}
\label{fig:kout_ex}
\end{minipage}
\begin{minipage}{.5\textwidth}
\centering
\includegraphics[width=\textwidth]{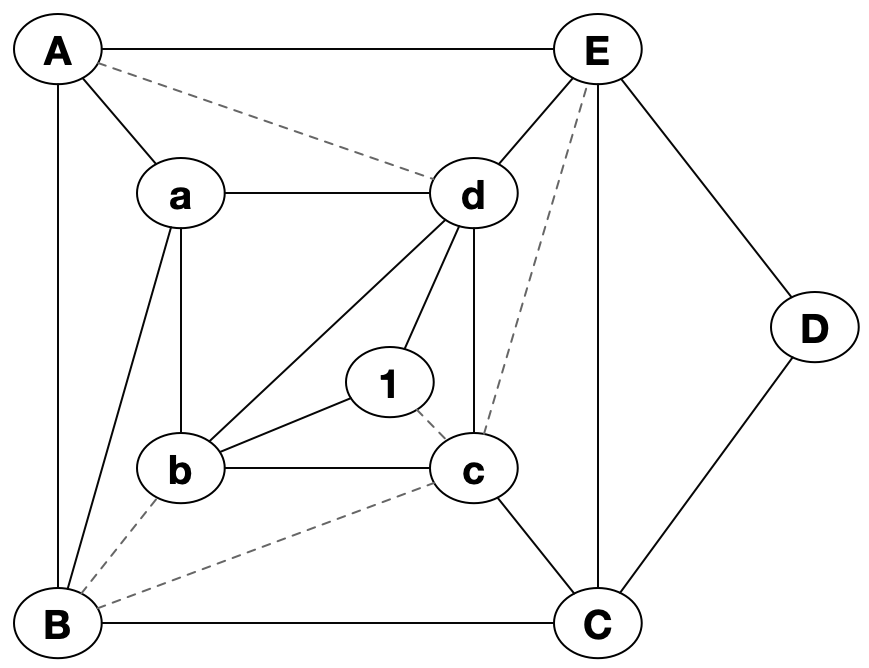}
\caption{figure}{Triangulation of example}
\label{fig:kout_ex_tri}
\end{minipage}
\end{figure}

\subsubsection{Tree Construction}
\label{koutTrees}
Since a level $i$ component is an outerplanar graph, we can construct a tree for each level $i$ component where the vertices of the tree represent interior faces and exterior edges of the component, just as in the previous section. However, we need to restrict how the roots and leftmost children are chosen for trees of level $i$ components, where $i > 1$.

Suppose we want to construct the tree for the level $i$ component $C$ that is enclosed by a level $i-1$ face $f$. Furthermore, suppose we have already constructed the tree for the level $i-1$ component that contains $f$, hence $f$ is represented by some vertex $(x, y)$ in that tree. The root of the tree for $C$ will be of the form $(z, z)$ for some vertex $z$ in $C$. If $x = y$, we pick $z$ so that $z$ is adjacent to $x$ in the triangulation of the graph. If $x \neq y$, we pick $z$ to be the third vertex in the triangle containing $x$ and $y$ in the triangulation.

The leftmost child for $(z, z)$ (assuming $C$ contains more than one vertex) is of the form $(z, u)$, where $u$ is the vertex in $C$ such that $(z, u)$ is an exterior edge in $C$ and is the first edge counterclockwise around $z$ after the edge $(z, x)$ in the triangulation. With these two restrictions, the rest of the tree is constructed exactly as in the previous section. The trees for the components of the graph in Figure \ref{fig:kout_ex} are given in Figure \ref{fig:kout_trees}.

\begin{figure}[h]
\centering
\includegraphics[width=\textwidth]{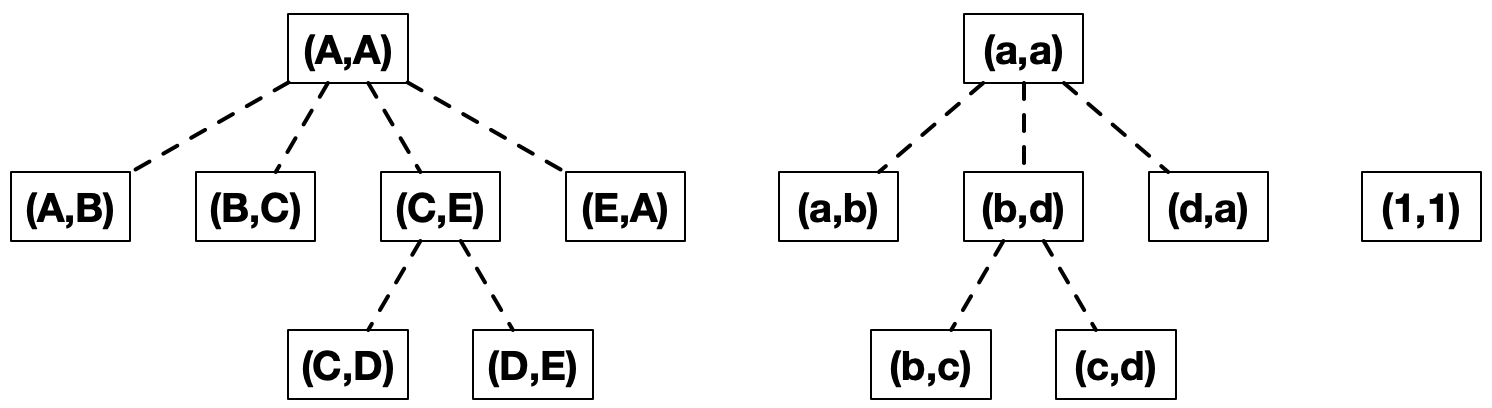}
\caption{figure}{Trees for example graph}
\label{fig:kout_trees}
\end{figure}

\subsubsection{Slice Construction}
\label{koutSlices}
The dynamic programming solution follows a divide and conquer paradigm, where we divide the graph into so-called \textit{slices} and calculate the tables for each slice before merging them together. Each vertex in each tree (and hence each interior face and exterior edge in each component of the graph) corresponds to a particular slice. We lay out the construction of these slices below. The idea is to first define left and right \textit{boundaries} for each tree vertex, and then define their slices by taking the induced subgraph of all vertices and edges that exist between the boundaries. The analogy here is that you can obtain a slice of pie by first deciding where your two cut lines will be (the left and right boundaries), and then your slice will be the pie that exists between your cuts.

Let $C$ be a level $i$ component enclosed by the level $i-1$ face $f$, and let $(x,y)$ be the tree vertex corresponding to $f$. For a particular slice of $C$, we want the left and right boundaries to each contain one vertex from each level of the graph, starting at level $i$ and extending outward to level 1. We will choose these vertices to be \textit{dividing points}. Given two successive leaf vertices $(x_1, x_2), (x_2, x_3)$ of the tree with root corresponding to $C$, we say a vertex $y$ in $f$ is a \textit{dividing point} for $(x_1, x_2)$ and $(x_2, x_3)$ if the edges $(x_2, x_1), (x_2, y), (x_2, x_3)$ occur in counterclockwise order around $x_2$ in the triangulation of the graph. For example, in Figure \ref{fig:kout_ex_tri}, $B$ is a dividing point for $(a,b)$ and $(b,c)$, and $B$, $C$, and $E$ are dividing points for $(b,c)$ and $(c,d)$.

We now define a pair of functions $\lbn$ and $\rbn$, which assign to each tree vertex a \textit{left boundary number} and a \textit{right boundary number}, respectively. Note that these boundary numbers will only be assigned to trees corresponding to level $i$ components for $i > 1$. This is because we will use boundary numbers to recursively define the boundaries themselves, where level 1 boundaries are the base case and level $i > 1$ boundaries are obtained from level $i-1$ boundaries. Using the same setup for $C$ and $f$ as given above, we define these functions as follows:

\begin{enumerate}[(F1)]
\item Let the leaves of the tree corresponding to $C$ be $v_1, v_2, \dots, v_t$ from left to right, where $v_i$ has label $(x_i, x_{i+1})$ for $1 \leq i \leq t$. Let the children of the tree vertex corresponding to $f$ be $z_1, z_2, \dots, z_r$ from left to right, where $z_j$ has label $(y_j, y_{j+1})$ for $1 \leq j \leq r$. Set $\lbn(v_1) = 1$ and $\rbn(v_t) = r + 1$. For $1 < i \leq t$, define $\lbn(v_i) = q$, where $q$ is the least $p \geq \lbn(v_{i-1})$ such that $y_p$ is a dividing point for $(x_{i-1}, x_i)$ and $(x_i, x_{i+1})$. For $1 \leq i < t$, define $\rbn(v_i) = \lbn(v_{i+1})$.
\item Let $v$ be a non-leaf vertex in the tree corresponding to $C$, and suppose $v$ has leftmost child $c_L$ and rightmost child $c_R$. Define $\lbn(v) = \lbn(c_L)$ and $\rbn(v) = \rbn(c_R)$.
\end{enumerate}

For example, let $C$ be the level 2 component in Figure \ref{fig:kout_ex} represented by the tree with root $(a,a)$. The enclosing face $f$ is represented by the tree vertex $(A,A)$. We first assign $\lbn((a,b)) = 1$ and $\rbn((d,a)) = 5$ since $(A,A)$ has 4 children. As mentioned above, $B$ is the only dividing point for $(a,b)$ and $(b,c)$, so we set $\lbn((b,c)) = 2$ (and hence $\rbn((a,b)) = 2$). We know that $B$, $C$, and $E$ are dividing points for $(b,c)$ and $(c,d)$, with indices 2, 3, and 4, respectively, out of the list of vertices in the children of $(A,A)$. Since 2 is the least $p \geq \lbn((b,c)) = 2$, we set $\lbn((c,d)) = 2$ (and hence $\rbn((b,c)) = 2$).

A complete list of boundary numbers can be seen in Figure \ref{fig:kout_trees_bns}, where we write the left and right boundary numbers on the left and right side of each tree vertex, respectively. Note that $\lbn((x_{i-1},x_i)) = \rbn((x_i,x_{i+1}))$ for every pair of consecutive tree vertices $(x_{i-1},x_i)$ and $(x_i,x_{i+1})$.

\begin{figure}[h]
\centering
\includegraphics[width=.7\textwidth]{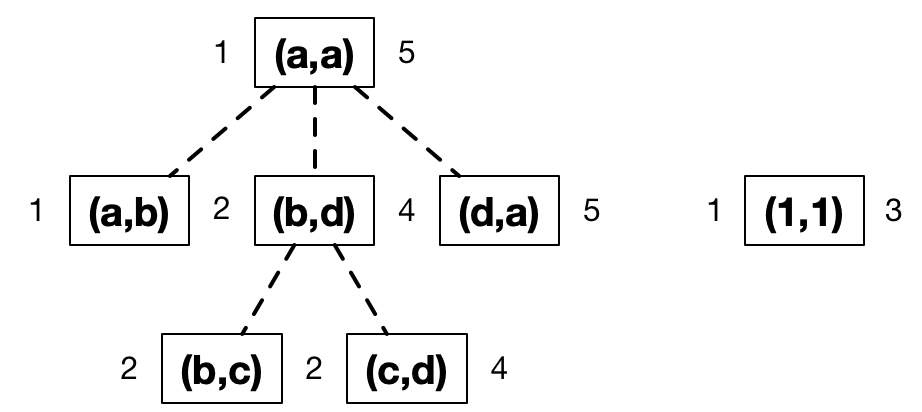}
\caption{figure}{Trees with boundary numbers}
\label{fig:kout_trees_bns}
\end{figure}

We are now ready to formally define left and right boundaries. Let $v$ be a tree vertex with label $(x,y)$. Then,

\begin{enumerate}[(B1)]
\item If $v$ is a vertex in a level 1 tree, the left boundary of the slice of $v$ is $x$ and the right boundary is $y$.
\item Otherwise, $v$ is a vertex in a level $i$ tree that corresponds to some level $i$ component enclosed by a level $i-1$ face $f$. Suppose the tree vertex corresponding to $f$, denoted $\vertex(f)$, has $s$ children. Let $\ell = \lbn(v)$ and let $r = \rbn(v)$. Then the left boundary of the slice of $v$ is $x$ plus the left boundary of the slice of the $\ell$th child of $\vertex(f)$, where we can let the left boundary of the nonexistent $(s+1)$ child of $\vertex(f)$ equal the right boundary of the $s$ child of $\vertex(f)$. Similarly, the right boundary of the slice of $v$ is $y$ plus the right boundary of the $(r-1)$ child of $\vertex(f)$, where we can let the right boundary of the nonexistent 0th child of $\vertex(f)$ equal the left boundary of the 1st child of $\vertex(f)$.
\end{enumerate}

A complete list of boundaries for every tree vertex is given in Figure \ref{fig:kout_trees_bndrs}. As with boundary numbers, the left and right boundaries have been written on the left and right side of each corresponding tree vertex. By construction, the right boundary of the tree vertex $(x_{i-1},x_i)$ must equal the left boundary of the next consecutive tree vertex $(x_i,x_{i+1})$.

\begin{figure}[h]
\centering
\includegraphics[width=\textwidth]{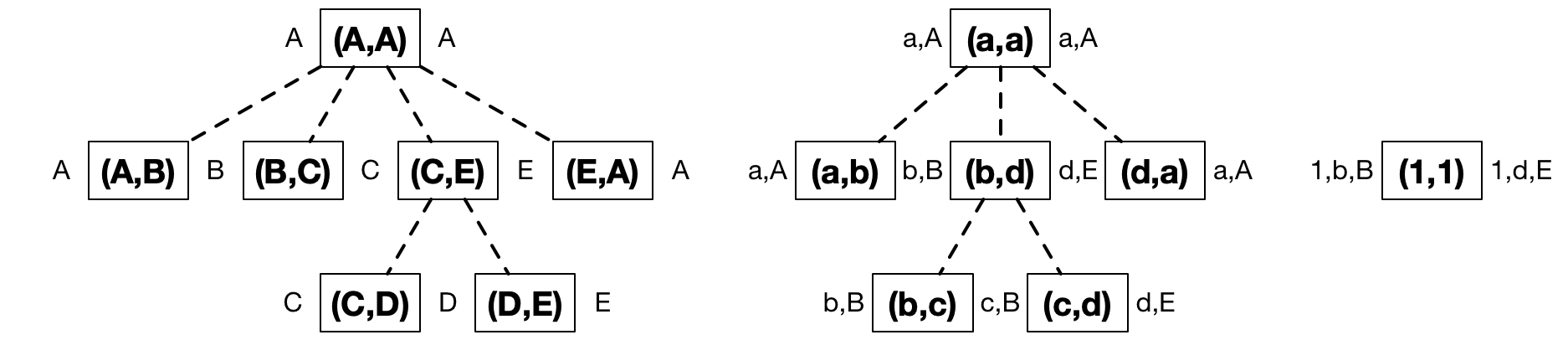}
\caption{figure}{Trees with slice boundaries}
\label{fig:kout_trees_bndrs}
\end{figure}

With the definition for boundaries in place, we can now formally define slices. Let $v$ be a level $i$ tree vertex for $i \geq 1$ with label $(x,y)$. We denote the slice of $v$ as $\slice(v)$, which is defined as follows:

\begin{enumerate}[(S1)]
\item If $v$ represents a level $i$ face $f$ that does not enclose any level $i+1$ components, then $\slice(v)$ is the union of the slices of the children of $v$, together with $(x,y)$ if $(x,y)$ is an edge.
\item If $v$ represents a level $i$ face $f$ that encloses a level $i+1$ component $C$, then $\slice(v)$ is the slice of the root of the tree that represents $C$, together with $(x,y)$ if $(x,y)$ is an edge.
\item If $v$ is a level 1 leaf, then $\slice(v)$ is simply the edge $(x,y)$.
\item Suppose $v$ is a level $i$ leaf for $i > 1$. Let $f$ be the level $i-1$ face that encloses $(x,y)$, and let $u_j$ for $1 \leq j \leq t$ be the children of $\vertex(f)$ (the tree vertex representing $f$). Let each $u_j$ have label $(z_j, z_{j+1})$.
\begin{enumerate}[(i)]
   \item If $\lbn(v) \neq \rbn(v)$, then $\slice(v)$ is the union of each $\slice(u_j)$ for $\lbn(v) \leq j < \rbn(v)$, along with any edges from $x$ or $y$ to $z_j$ for $\lbn(v) \leq j \leq \rbn(v)$ and the edge $(x,y)$.
   \item If $\lbn(v) = \rbn(v) = b$, then $\slice(v)$ is the left boundary of $\slice(u_b)$ (including edges between consecutive boundary vertices), along with any edges from $x$ or $y$ to $z_b$ and the edge $(x,y)$.
\end{enumerate}
\end{enumerate}

The slices for $(a,b)$, $(b,d)$, and $(d,a)$ are given in Figure \ref{fig:slice_ex}, where boundary vertices are connected by dashed lines if they are not connected already. Since $(a,b)$ is a level 2 leaf enclosed by the face with tree vertex $(A,A)$, and $\lbn((a,b)) = 1 \neq 2 = \rbn((a,b))$, we see that $\slice((a,b))$ is equal to $\slice((A,B))$ along with the edges connecting $a$ to $A$ and $B$, where $\slice((A,B))$ is the edge $(A,B)$. Similarly, $\slice((d,a))$ is equal to $\slice((E,A))$ along with the edges  connected $d$ and $a$ to $E$ and $A$. Since $(b,d)$ is a level 2 tree vertex representing a face that encloses the level 3 component represented by $(1,1)$, its slice is equal to $\slice((1,1))$ along with the edge $(b,d)$. The slice for $(1,1)$ is obtained by taking the union of $\slice((b,c))$ and $\slice((c,d))$ and adding the edges connecting 1 to $b$ and $d$.

\begin{figure}[h]
\centering
\includegraphics[width=.8\textwidth]{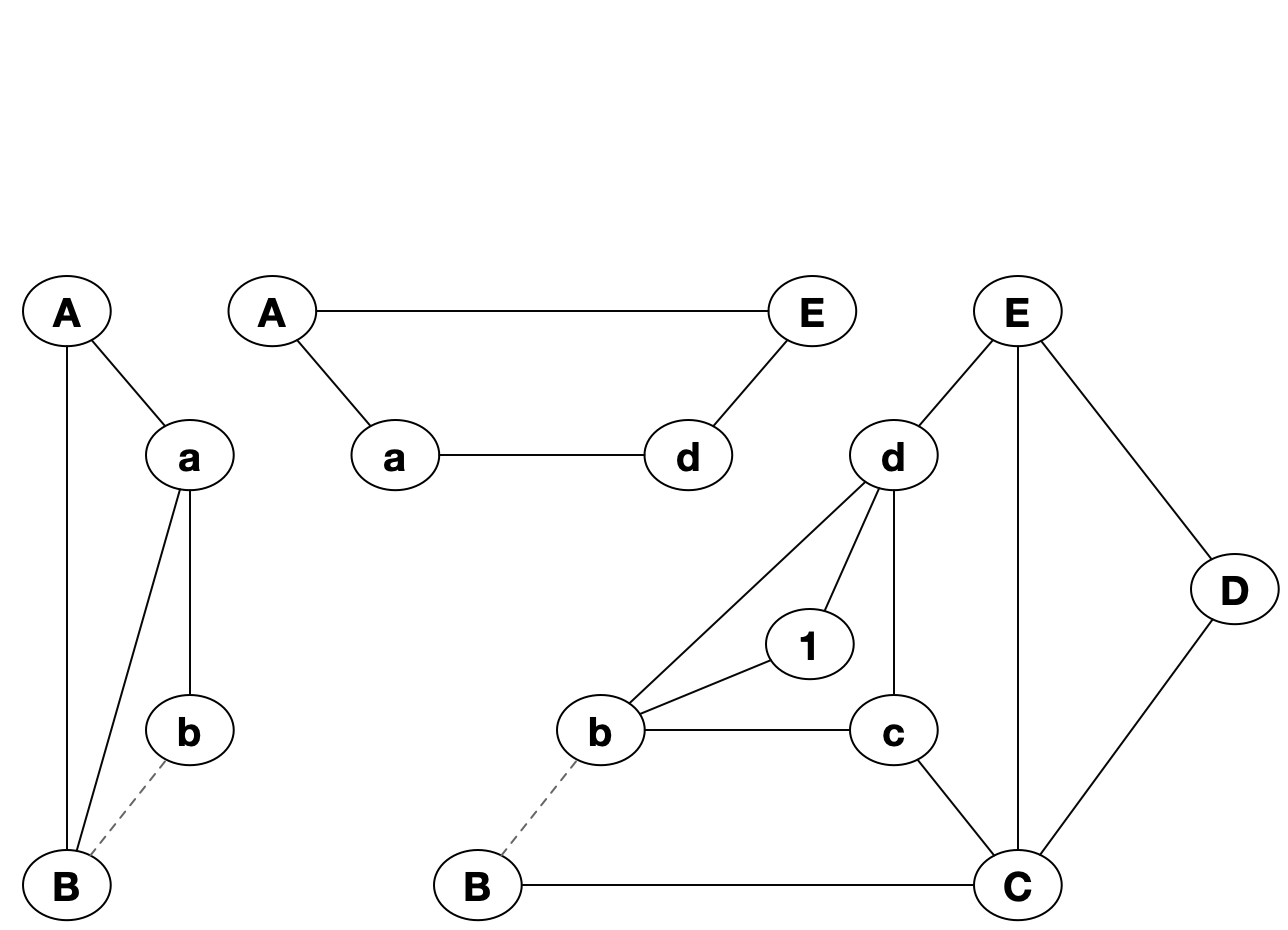}
\caption{figure}{Slices for $(a,b)$, $(b,d)$, and $(d,a)$.}
\label{fig:slice_ex}
\end{figure}

\subsubsection{Dynamic Program}
In this section, we detail the dynamic program that solves the densest $k$ subgraph problem for $b$-outerplanar graphs. Note that \textit{adjust, merge, extend} and \textit{contract} are original whereas the \textit{table} procedure is given by Baker~\cite{baker1994}. The pseudocode for these procedures can be found in Appendix~\ref{app:b_outerplanar_algs}. The program essentially follows the same recursive structure that we used to build the slices since it creates a table for each slice. The table for a level $i$ slice consists of $2^{2i}$ entries: one entry for each subset of the boundary vertices (recall that the left and right boundaries for a level $i$ slice contain exactly $i$ vertices each, for a total of $2i$ boundary vertices). An entry contains a number for each value of $k' = 0 \dots k$, where this number is the maximum number of edges over all subgraphs of the slice of exactly $k'$ vertices that contain the corresponding subset of the boundary vertices.

The main procedure of the program is \textit{table} (shown in Figure \ref{alg:table}), which takes as input a tree vertex $v = (x,y)$. This procedure contains a conditional branch corresponding to each of the rules S1 -- S4. The branch corresponding to S1 handles the case when $v$ represents a face $f$ that does not enclose a level $i+1$ component. In this case, the procedure makes a recursive call to \textit{table} on each child of $v$ and merges the resulting tables together. Note that this is analogous to how rule S1 creates $\slice(v)$ by taking the union of the slices for the children of $v$.

The second conditional branch (corresponding to rule S2) handles the case when $v$ represents a face $f$ that encloses a level $i+1$ component $C$. In this case, the procedure makes a recursive call to \textit{table} on the tree vertex that represents $C$. The resulting table is then passed to the \textit{contract} function, which turns a level $i+1$ table into a level $i$ table by removing the level $i+1$ boundary vertices from the table, and the contracted table is then returned.

The third conditional branch (corresponding to rule S3) handles the case when $v$ is a level 1 leaf. In this case, the procedure returns a template table that works for all level 1 leaf vertices, since any level 1 leaf represents a level 1 exterior edge of the graph.

The fourth conditional branch, which corresponds to rule S4, is slightly more complicated. This branch handles the case when $v$ is a level $i>1$ leaf vertex. The idea is to break up the slice of $v$ into subslices, compute the table for an initial subslice, and extend this table by merging the tables for the subslices clockwise and counterclockwise from the initial subslice. These subslices have their own respective subboundaries. In the second case of rule S4, where $\lbn(v) = \rbn(v)$, no subslices can be created, so the procedure will return the table for the whole slice by passing the vertex to the \textit{create} procedure, which effectively applies brute force to create the table for the subslice determined by the second parameter $p$ (this is explained in more detail below). 

In the case that $\lbn(v) \neq \rbn(v)$, we can create tables for subslices. Since we are dealing with a planar graph, there exists a level $i-1$ vertex $z_p$ such that all level $i-1$ vertices other than $z_p$ that are adjacent to $x$ are clockwise from $z_p$, and all level $i-1$ vertices other than $z_p$ that are adjacent to $y$ are counterclockwise from $z_p$. Here, $z_p$ is the only level $i-1$ vertex in $\slice(v)$ that can be adjacent to both $x$ and $y$ (although it might not be adjacent to either). So, we construct an initial level $i$ table for the subslice corresponding to $z_p$ using \textit{create}, and then we make as many calls as necessary to the \textit{merge} and \textit{extend} procedures (described below) to extend the table on one side with subslices constructed from the vertices adjacent to $x$, and on the other side with subslices constructed from the vertices adjacent to $y$.

\begin{figure}
\begin{minipage}{.45\textwidth}
\centering
\includegraphics[width=\textwidth]{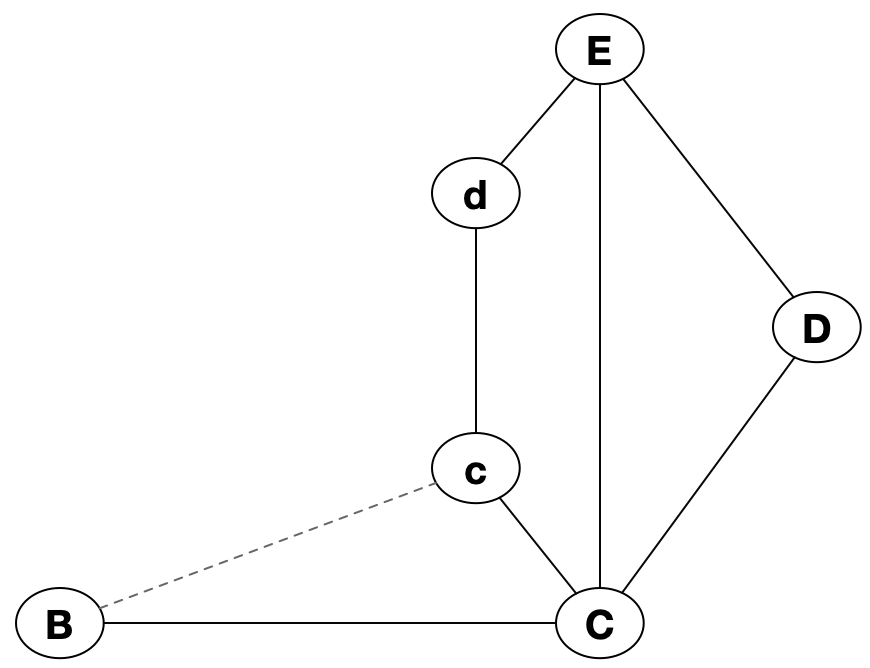}
\end{minipage}
\begin{minipage}{.55\textwidth}
\centering
\includegraphics[width=\textwidth]{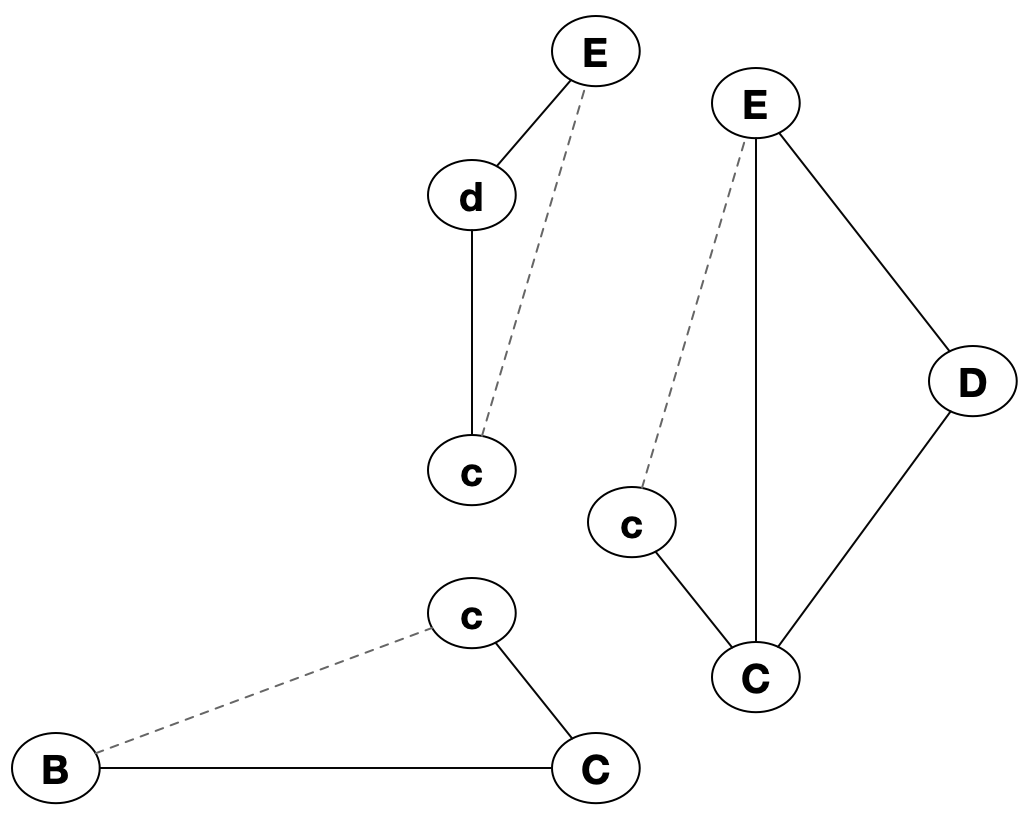}
\end{minipage}
\caption{figure}{The slice for $(c,d)$ is split into subslices for computing tables.}
\label{fig:cd_subslices}
\end{figure}

For example, Figure \ref{fig:cd_subslices} shows how we split $\slice((c,d))$ into subslices. The algorithm finds that $E$ is a level 1 vertex such that the level 1 vertices adjacent to $c$ are clockwise from $E$ and the level 1 vertices adjacent to $d$ are counterclockwise to $E$ (of which there are none). A call to \textit{create} constructs the table for the initial subslice with subboundaries $c,E$ and $d,E$. We then merge the other subslices in a clockwise fashion, first merging the table for the subslice with subboundaries $c,C$ and $c,E$, and then merging the table for the subslice with subboundaries $c,C$ and $c,B$.

The \textit{adjust} procedure, given in Figure \ref{alg:adjust}, takes as input a table $T$, which represents a slice with left boundary $L$ and right boundary $R$. Let $x$ and $y$ be the highest level boundary vertices in $L$ and $R$, respectively. This procedure checks if $x \neq y$, and if so, adds 1 to the table entry where $x$ and $y$ are both included. Unlike in Baker's original procedure, the case when $x = y$ is not handled in this procedure and is instead handled in \textit{merge}.


The \textit{merge} procedure, given in Figure \ref{alg:merge}, takes as input two tables $T_1$ and $T_2$ such that the right boundary of the slice that $T_1$ represents is the same as the left boundary of the slice that $T_2$ represents. Let the left and right boundaries of the slice that $T_1$ represents be $L$ and $M$, and let the left and right boundaries of the slice that $T_2$ represents be $M$ and $R$. The resulting table $T$ returned by \textit{merge} has left and right boundaries $L$ and $R$. This procedure constructs $T$ by creating an entry for each subset $A$ of $L \cup R$. As stated earlier, each entry contains a number for each value of $k'$, where $k'$ ranges from 0 to the number of vertices in the union of the two slices represented by $T_1$ and $T_2$, and where the number is the maximum number of edges over all subgraphs of exactly $k'$ vertices containing $A$. Each of these individual subgraphs contains a different subset $B$ of $M$.


The \textit{contract} procedure, given in Figure \ref{alg:contract}, changes a level $i+1$ table $T$ into a level $i$ table $T'$. Here, $T$ represents the slice for $(z,z)$, where $(z,z)$ is the root of a tree corresponding to a level $i+1$ component $C$ contained in a level $i$ face $f$, and $T'$ is the table for the slice of $\vertex(f)$. Let $S = \slice((z,z))$ and let $S' = \slice(\vertex(f))$. Let the left and right boundaries of $S'$ be $L$ and $R$, respectively. Then the left and right boundaries of $S$ are of the form $z, L$ and $z, R$ respectively. For each subset of $L \cup R$ and for each value of $k'$, $T$ contains two numbers: one that includes $z$ and one that does not include $z$. So, for each subset $A$ of $L \cup R$ and each value of $k'$, we set $T(A, k')$ equal to the larger of these two numbers.


The \textit{create} procedure takes as input a leaf vertex $v = (x,y)$ in a tree that corresponds to a level $i+1$ component enclosed by a face $f$, and a number $p \leq t+1$, where the children of $\vertex(f)$ are $u_1, u_2, \dots, u_t$. This procedure simply applies brute force to create the table for the subgraph containing the edge $(x,y)$, the subgraph induced by the left boundary of $u_p$ if $p \leq t$ or the right boundary of $u_{p-1}$ if $p = t + 1$, and any edges from $x$ or $y$ to the level $i$ vertex of this boundary (this construction is similar to the construction of slices, hence the term subslice is appropriate). Due to the simplicity of this function, it will not be explicitly written here.

Lastly, the \textit{extend} procedure, given in Figure \ref{alg:extend}, takes as input a level $i+1$ vertex $z$ and a table $T$ representing a level $i$ slice, and produces a table $T'$ for a level $i+1$ slice. Let $L$ and $R$ be the boundaries for the level $i$ slice represented by $T$. The boundaries for the new slice will be $L \cup \{z\}$ and $R \cup \{z\}$. For each subset $A$ of $L \cup R$ and each value of $k'$, the new table $T'$ contains two values: one that includes $z$ and one that does not include $z$. The entries $T'(A,k')$ which do not include $z$ can simply be set to their original values in $T$. For the entries $T'(A \cup \{z\},k')$ that do include $z$, we first check that $T(A,k'-1)$ is not undefined. If this is the case, we set $T'(A \cup \{z\},k')$ as $T(A,k'-1)$ plus the number of edges between $z$ and every vertex in $A$.


We claim that calling the above algorithm on the root of the level 1 tree results in a correct table for the slice of the root and that this slice is actually the entire graph. Since the level 1 root is of the form $(x,x)$, its left and right boundaries are both equal to $x$. Thus, the table for this root has exactly 4 rows, and two of these rows are invalid since they attempt to include one copy of $x$ and not the other (which is nonsensical). The two remaining rows have numbers corresponding to the maxmimum number of edges in subgraphs of size exactly $k'$ for $k' = 0,\dots,k$. Taking the maximum between the two numbers corresponding to when $k' = k$ gives us the size of the densest $k$ subgraph.

\subsubsection{Proof of Correctness}
Baker's original paper provides us with a proof that the procedure \textit{table} is correct under the assumption that the procedures \textit{adjust}, \textit{merge}, \textit{contract}, \textit{create}, and \textit{extend} are implemented correctly. While we have to slightly modify this proof so that it works for our implementation, the bulk of the work is in proving that the above five procedures are in fact correct for the densest $k$ subgraph problem. Since \textit{create} is just a brute-force algorithm, its implementation is trivial and is omitted from the proof.

\begin{lemma}
The procedures \textit{adjust}, \textit{merge}, \textit{contract}, and \textit{extend} defined above are implemented correctly.
\end{lemma}

\begin{proof}
First, we show that a call to the  \textit{adjust} procedure on a table $T$ results in a correctly adjusted table $T'$; that is, $T'$ accounts for when the top level boundary vertices of the corresponding slice are equal or not. Let $L$ and $R$ be the left and right boundaries, respectively, for the sice that $T$ represents, and let $x$ and $y$ be the highest level vertices in $L$ and $R$, respectively. Since the case when $x = y$ is already handled in the \textit{merge} procedure, the \textit{adjust} procedure only needs to check when $x \neq y$, and in particular when $(x,y)$ is an edge in the graph. In this case, any entry in $T$ that includes both $x$ and $y$ should be incremented by one in order to count this edge. Thus, \textit{adjust} is correct.

Now, we show that a call to the \textit{merge} procedure on two tables $T_1$ and $T_2$ representing slices $S_1$ and $S_2$, respectively, results in a correct table $T$ representing the slice $S = S_1 \cup S_2$. Let $L$ and $M$ be the left and right boundaries of $S_1$ and let $M$ and $R$ be the left and right boundaries of $S_2$, where we know the right boundary of $S_1$ is the same as the left boundary of $S_2$ by assumption. Recall that a table has an entry for each subset of the boundaries of the slice it represents. This means $T_1$ has an entry for each subset of $L \cup M$, $T_2$ has an entry for each subset of $M \cup R$, and we want $T$ to have an entry for each subset of $L \cup R$. So, for a subset $A$ of $L \cup R$, we show that the \textit{merge} procedure calculates the maximum number of edges over all subgraphs of $S$ with exactly $k'$ vertices containing $A$, where $k'$ ranges from 0 to $\max_{T_1}k' + \max_{T_2}k' - |M|$. This upper bound is obtained by adding the number of vertices in $S_1$ and $S_2$ (which are given by the largest values of $k'$ in $T_1$ and $T_2$) and subtracting the number of vertices shared by the two slices, which by definition are the vertices in $M$. Now, fix $k'$. In order to loop over every subgraph of $S$ containing $A$ with exactly $k'$ vertices, the \textit{merge} procedure essentially loops over every pair of subgraphs of $S_1$ and $S_2$ such that the number of vertices in the union of these two subgraphs is exactly $k'$. This is accomplished by looping over every subset $B$ of $M$, and referring to the table entries in $T_1$ and $T_2$ corresponding to the subsets $(A \cap L) \cup B$ and $(A \cap R) \cup B$. Fix the subset $B$ of $M$. Extra care needs to be taken in determining which values of $k_1$ and $k_2$ can be used in $T_1$ and $T_2$, respectively, so that the union of the corresponding subgraphs of $S_1$ and $S_2$ contains exactly $k'$ vertices. The number of vertices in this union should be $k_1 + k_2 - |B|$ since the two subgraphs only share the vertices in $B$. However, we need to account for when the top level vertices in $L$ and $R$ are equal; the procedure handles this by subtracting 1 from the above vertex count if this is the case. Then for particular values of $k_1$ and $k_2$, the procedure keeps track of the number of edges in the described union, which is obtained by adding the table values $T_1((A \cap L) \cup B, k_1)$ and $T_2((A \cap R) \cup B, k_2)$, and then subtracting the number of double counted edges, which are edges between vertices in $B$. After calculating all of these values, the procedure sets $T(A, k')$ to be the maximum such value, which must be the maximum number of edges in a subgraph containing $A$ with exactly $k'$ vertices. If all of the values were undefined, the procedure sets $T(A, k')$ to be undefined. Thus, the \textit{merge} procedure is correct.

Next, we show that a call to the \textit{contract} procedure on a level $i+1$ table $T$ results in a correct level $i$ table $T'$. Here, $T$ represents the slice $S$ for a level $i+1$ component $C$ contained in the level $i$ face $f$, and we want $T'$ to represent the slice $S'$ for $f$. Let $L$ and $R$ be the left and right boundaries, respectively, for $S'$, and let $(z,z)$ be the label for the root of the tree corresponding to $C$. Note that the left and right boundaries for $S$ are by definition $L \cup \{z\}$ and $R \cup \{z\}$, respectively. So, we want to show that \textit{contract} has a correct entry for each subset of $L \cup R$. For each subset $A$ of $L \cup R$ and for each value of $k'$, the \textit{contract} procedure sets $T'(A, k')$ to be the maximum value between $T(A, k')$ and $T(A \cup \{z\}, k')$. Note that $T(A, k')$ is by definition the maximum number of edges over all subgraphs of $S$ containing $A$ with exactly $k'$ vertices, and $T(A \cup \{z\}, k')$ is the maximum number of edges over all subgraphs of $S$ containing $A$ as well as $z$ with exactly $k'$ vertices. The maximum between these two values must then be the maximum number of edges over all subgraphs of $S'$ containing $A$ with exactly $k'$ vertices. Thus, the \textit{contract} procedure is correct.

Lastly, we show that a call to the \textit{extend} procedure on a table $T$ representing a level $i$ slice $S$ and a level $i+1$ vertex $z$ results in a correct table $T'$ representing the slice $S \cup \{z\}$. Let $L$ and $R$ be the left and right boundaries for $S$, respectively. We will show that \textit{extend} produces correct entries for each subset of $L \cup R \cup \{z\}$. Clearly, the entry in $T'$ for a subset that does not include $z$ should equal the entry in $T$ for the same subset, which is precisely how \textit{extend} initializes $T'$. So, we only need to focus on the entries that include $z$. Let $A$ be a subset of $L \cup R$, and fix $k'$. The procedure first checks that $T(A, k'-1)$ is not undefined, for if it were undefined, this would mean that there is no value for a subgraph of $S$ containing $A$ with exactly $k'-1$ vertices, and hence there cannot be a value for the same subgraph after adding in $z$ with exactly $k'$ vertices. We then let $m$ be the number of edges between $z$ and every vertex in $A$. The maximum number of edges over all subgraphs of $S'$ containing $A \cup \{z\}$ with exactly $k'$ vertices should then be the maximum number of edges over all subgraphs of $S$ containing $A$ with exactly $k'-1$ vertices, plus $m$. This is exactly what \textit{extend} sets as the value for $T'(A \cup \{z\}, k')$, hence the procedure must be correct.
\end{proof}

\begin{theorem}
A call of \textit{table} on the root of the level 1 tree results in a correct table for the corresponding slice, which is in fact the entire graph.
\end{theorem}

\begin{proof}
The argument is essentially the same as Lemma 2 in \cite{baker1994}, which focuses on the maximum independent set problem. The only difference is that the \textit{adjust} procedure in the maximum independent set scenario deals with the case $x = y$, where $(x,y)$ is the label of the input tree vertex $v$, whereas our implementation for densest $k$ subgraph deals with this case in the \textit{merge} procedure. Since this case is handled regardless, the proof still follows.
\end{proof}

\subsubsection{Performance Bounds}
Constructing the trees and computing the boundaries of slices requires linear time. By Lemma 1 in \cite{baker1994}, the \textit{table} procedure is called once for each tree vertex. Recall that for a level $i$ tree, a leaf vertex corresponds to an exterior edge of a level $i$ face, and a tree vertex which is neither a leaf nor the root corresponds to a level $i$ face, which must have at least one interior edge. The root of a level $i$ tree corresponds to a level $i$ component, which for $i > 1$ must be associated with at least one edge connecting the level $i$ component to the level $i-1$ component (for the level 1 root, there may be no associated edge). Thus, the number of tree vertices over all trees is at most the number of edges in the $b$-outerplanar graph (plus one to account for the level 1 root). Since the number of edges in a planar graph is linear in the number of vertices, the number of calls on \textit{table} is linear in the number of vertices as well.

Each table constructed in the dynamic program has one row for each subset of the left and right boundary vertices, of which there can be at most $2b$ for a $b$-outerplanar graph; that is, $b$ vertices in each of the left and right boundaries. So, there are at most $4^b$ rows in a table. In each row, there are at most $k$ entries for the given input $k$. Since each call to \textit{adjust}, \textit{extend}, or \textit{contract} requires iterating over a table, the time complexity for each is $O(k4^b)$. The \textit{merge} procedure has four nested loops, where the first (outermost) iterates at most $4^b$ times, the second at most $k$ times, the third at most $2^b$ times, and the fourth (innermost) at most $k$ times. Thus, the time complexity of \textit{merge} is $O(k^2 8^b)$. Lastly, the time complexity of \textit{create} must be $O(kb4^b)$ since it constructs an entire table by brute force. Each of these procedures is called at most once for each recursive call of \textit{table}. Therefore, the time complexity must be bounded by \textit{merge}, hence the time complexity of \textit{table} is $O(k^2 8^b n)$.

\section{Polynomial Time Approximation Scheme and Future Work}\label{sec:PTAS}
When searching for a polynomial time approximatin scheme (PTAS) for planar graph problems, one often attempts to use Baker's technique. For this technique, we assume that we have a dynamic programming solution to the given problem in the $b$-outerplanar case. This technique works as follows: Given a planar graph $G$ and a positive number $\epsilon$, let $b = \frac{1}{\epsilon}$. Perform a breadth-first search on $G$ to obtain a BFS tree $T$, and number the levels of $T$ starting from the root, which is level 0. For each $i = 0,1,\dots,b-1$, let $G_i$ be the subgraph of $G$ induced by the vertices on the levels of $T$ that are congruent to $i$ modulo $b$. Note that $G_i$ is likely disconnected. Let the connected components of $G_i$ be $G_{i,0}, G_{i,1}, \dots$. Since each $G_{i,j}$ is $(b-1)$-outerplanar by construction, we may run the given dynamic program on each $G_{i,j}$ and combine the solutions over all $j$ to obtain a solution $S_i$ for the graph $G_i$. We then take the maximum $S_i$, denoted $S$, as our approximate solution.

We hypothesize that this technique will not work for the densest $k$
subgraph problem on planar graphs. The reason is that by having a
potentially large number of disconnected components, the approximate
solution cannot be guaranteed to be within the bound given by
$\epsilon$. Suppose we have an approximate solution $S$ for the
densest $k$ subgraph problem on some planar graph $G$. Note that $S$
is the exact solution for the graph $G_i$ for some $i$, meaning $S$
does not account for any vertices on the levels of $T$ which are
congruent to $i$ modulo $b$. While $S$ could still be very dense, it
is possible that most (if not all) of the edges in $G$ are between
vertices in different levels of $T$. This allows for the possibility
that no matter which $S_i$ is chosen as the maximum, each graph $G_i$
is missing too many edges to closely approximate the optimal
solution. For future work, it would be of great interest for one to
prove that such a construction is impossible using Baker's technique.

\section{Acknowledgements}
We would like to express our sincere thanks to Samuel Chase for his
collaboration on our initial explorations of finding a PTAS for the
densest $k$ subgraph problem.

\bibliography{ptas.bib}

\appendix

\section{A Complete Set of Tables for the Example in Figure~\ref{fig:outerplanar_tree construction}}\label{app:outerplanar_tables}
As mentioned in Section~\ref{sec:outerplanar}, the tables for each leaf in the tree are the same. We list them here for completeness.

$Table(P) = T_{(c,b)} = $
\begin{tabular}{|l|l|l|l|l|}
\hline
$c$ & $b$ & $k=0$ & $k=1$ & $k=2$ \\ \hline
0   & 0   &   0                      &   $\emptyset$    &   $\emptyset$    \\ \hline
0   & 1   &   $\emptyset$    &    0                     &    $\emptyset$   \\ \hline
1   & 0   &   $\emptyset$    &    0                     &    $\emptyset$   \\ \hline
1   & 1   &   $\emptyset$    &    $\emptyset$   &   1 \\ \hline
\end{tabular}

\vspace{.75cm}

$Table(L) = T_{(b,a)} = $
\begin{tabular}{|l|l|l|l|l|}
\hline
$b$ & $a$ & $k=0$ & $k=1$ & $k=2$ \\ \hline
0   & 0   &   0                      &   $\emptyset$    &   $\emptyset$    \\ \hline
0   & 1   &   $\emptyset$    &    0                     &    $\emptyset$   \\ \hline
1   & 0   &   $\emptyset$    &    0                     &    $\emptyset$   \\ \hline
1   & 1   &   $\emptyset$    &    $\emptyset$   &   1 \\ \hline
\end{tabular}

\vspace{.75cm}

$Table(M) = T_{(a,e)} = $
\begin{tabular}{|l|l|l|l|l|}
\hline
$a$ & $e$ & $k=0$ & $k=1$ & $k=2$ \\ \hline
0   & 0   &   0                      &   $\emptyset$    &   $\emptyset$    \\ \hline
0   & 1   &   $\emptyset$    &    0                     &    $\emptyset$   \\ \hline
1   & 0   &   $\emptyset$    &    0                     &    $\emptyset$   \\ \hline
1   & 1   &   $\emptyset$    &    $\emptyset$   &   1 \\ \hline
\end{tabular}

\vspace{.75cm}

$Table(N) = T_{(e,f)} = $
\begin{tabular}{|l|l|l|l|l|}
\hline
$e$ & $f$ & $k=0$ & $k=1$ & $k=2$ \\ \hline
0   & 0   &   0                      &   $\emptyset$    &   $\emptyset$    \\ \hline
0   & 1   &   $\emptyset$    &    0                     &    $\emptyset$   \\ \hline
1   & 0   &   $\emptyset$    &    0                     &    $\emptyset$   \\ \hline
1   & 1   &   $\emptyset$    &    $\emptyset$   &   1 \\ \hline
\end{tabular}

\vspace{.75cm}

$Table(O) = T_{(f,g)} = $
\begin{tabular}{|l|l|l|l|l|}
\hline
$f$ & $g$ & $k=0$ & $k=1$ & $k=2$ \\ \hline
0   & 0   &   0                      &   $\emptyset$    &   $\emptyset$    \\ \hline
0   & 1   &   $\emptyset$    &    0                     &    $\emptyset$   \\ \hline
1   & 0   &   $\emptyset$    &    0                     &    $\emptyset$   \\ \hline
1   & 1   &   $\emptyset$    &    $\emptyset$   &   1 \\ \hline
\end{tabular}

\vspace{.75cm}

$Table(R) = T_{(g,d)} = $
\begin{tabular}{|l|l|l|l|l|}
\hline
$g$ & $d$ & $k=0$ & $k=1$ & $k=2$ \\ \hline
0   & 0   &   0                      &   $\emptyset$    &   $\emptyset$    \\ \hline
0   & 1   &   $\emptyset$    &    0                     &    $\emptyset$   \\ \hline
1   & 0   &   $\emptyset$    &    0                     &    $\emptyset$   \\ \hline
1   & 1   &   $\emptyset$    &    $\emptyset$   &   1 \\ \hline
\end{tabular}

\vspace{.75cm}

$Table(Q) = T_{(d,c)} = $
\begin{tabular}{|l|l|l|l|l|}
\hline
$d$ & $c$ & $k=0$ & $k=1$ & $k=2$ \\ \hline
0   & 0   &   0                      &   $\emptyset$    &   $\emptyset$    \\ \hline
0   & 1   &   $\emptyset$    &    0                     &    $\emptyset$   \\ \hline
1   & 0   &   $\emptyset$    &    0                     &    $\emptyset$   \\ \hline
1   & 1   &   $\emptyset$    &    $\emptyset$   &   1 \\ \hline
\end{tabular}

\vspace{.75cm}

$Table(H) = Merge(T_{(b,a)},T_{(a,e)},3) = T_{(b,e)} =$

\begin{tabular}{|l|l|l|l|l|l|}
\hline
$b$ & $e$ & $k=0$       & $k=1$       & $k=2$       & $k=3$       \\ \hline
0   & 0   & 0           & 0           & $\emptyset$ & $\emptyset$ \\ \hline
0   & 1   & $\emptyset$ & 0           & 1           & $\emptyset$ \\ \hline
1   & 0   & $\emptyset$ & 0           & 1           & $\emptyset$ \\ \hline
1   & 1   & $\emptyset$ & $\emptyset$ & 1           & 3           \\ \hline
\end{tabular}

\vspace{.75cm}

$Merge(H,N) = Merge(T_{(b,e)},T_{(e,f)},4) = T_{(b,f)} =$

\begin{tabular}{|l|l|l|l|l|l|l|}
\hline
$b$ & $f$ & $k=0$       & $k=1$       & $k=2$       & $k=3$       & $k=4$\\ \hline
0   & 0   & 0           & 0           & 1           & $\emptyset$ & $\emptyset$\\ \hline
0   & 1   & $\emptyset$ & 0           & 1           & 2           & $\emptyset$\\ \hline
1   & 0   & $\emptyset$ & 0           & 1           & 3           & $\emptyset$\\ \hline
1   & 1   & $\emptyset$ & $\emptyset$ & 0           & 2           & 4           \\ \hline
\end{tabular}

\vspace{.75cm}

$Table(I) = Merge(T_{(b,f)},T_{(f,g)},5) = T_{(b,g)} =$

\begin{tabular}{|l|l|l|l|l|l|l|l|}
\hline
$b$ & $g$ & $k=0$       & $k=1$       & $k=2$       & $k=3$       & $k=4$        & $k=5$\\ \hline
0   & 0   & 0           & 0           & 1           & 2           & $\emptyset$  & $\emptyset$\\ \hline
0   & 1   & $\emptyset$ & 0           & 1           & 2           & 3            & $\emptyset$\\ \hline
1   & 0   & $\emptyset$ & 0           & 1           & 3           & 4            & $\emptyset$\\ \hline
1   & 1   & $\emptyset$ & $\emptyset$ & 1           & 2           & 4            & 6 \\ \hline
\end{tabular}

\vspace{.75cm}

$Table(K) = Merge(T_{(g,d)},T_{(d,c)},3) = T_{(g,c)} =$

\begin{tabular}{|l|l|l|l|l|l|}
\hline
$g$ & $c$ & $k=0$       & $k=1$       & $k=2$       & $k=3$       \\ \hline
0   & 0   & 0           & 0           & $\emptyset$ & $\emptyset$ \\ \hline
0   & 1   & $\emptyset$ & 0           & 1           & $\emptyset$ \\ \hline
1   & 0   & $\emptyset$ & 0           & 1           & $\emptyset$ \\ \hline
1   & 1   & $\emptyset$ & $\emptyset$ & 1           & 3           \\ \hline
\end{tabular}

\vspace{.75cm}

$Merge(P,I) = Merge(T_{(c,b)}, T_{(b,g)},6) = T_{(c,g)} =$

\begin{tabular}{|l|l|l|l|l|l|l|l|l|}
\hline
$c$ & $g$ & $k=0$       & $k=1$       & $k=2$       & $k=3$       & $k=4$       & $k=5$           & $k=6$       \\ \hline
0   & 0   & 0           & 0           & 1           & 3           & 4           & $\emptyset$     & $\emptyset$            \\ \hline
0   & 1   & $\emptyset$ & 0           & 1           & 2 &         3             & 6               & $\emptyset$  \\ \hline
1   & 0   & $\emptyset$ & 0           & 1           & 2           & 4           & 5               & $\emptyset$  \\ \hline
1   & 1   & $\emptyset$ & $\emptyset$ & 1           & 3           & 4           & 6               & 8            \\ \hline
\end{tabular}

\vspace{.75cm}

$Table(J) = Merge(T_{(c,g)},T_{(g,c)},7)  = T_{(c,c)} =$

\begin{tabular}{|l|l|l|l|l|l|l|l|l|l|}
\hline
$c$ & $c$ & $k=0$       & $k=1$       & $k=2$       & $k=3$       & $k=4$       & $k=5$       & $k=6$       & $k=7$       \\ \hline
0   & 0   & 0           & 0           & 1           & 3           & 4           & 6           & 7           & $\emptyset$ \\ \hline
0   & 1   & $\emptyset$ & $\emptyset$ & $\emptyset$ & $\emptyset$ & $\emptyset$ & $\emptyset$ & $\emptyset$ & $\emptyset$ \\ \hline
1   & 0   & $\emptyset$ & $\emptyset$ & $\emptyset$ & $\emptyset$ & $\emptyset$ & $\emptyset$ & $\emptyset$ & $\emptyset$ \\ \hline
1   & 1   & $\emptyset$ & 0           & 1           & 3           & 4           & 6           & 8           & 10          \\ \hline
\end{tabular}

\newpage
\section{Procedures for the Exact $b$-outerplanar Algorithm}\label{app:b_outerplanar_algs}
\begin{figure}[!htbp]
\begin{algorithmic}
\small
\Procedure{table}{$v$}
   \State let $(x,y)$ be the label of $v$;
   \State let $i$ be the level of $v$;
   \If{$v$ is a vertex representing a face $f$ and $f$ encloses no level $i+1$ component}
      \State $T = \algtable(u)$, where $u$ is the leftmost child of $v$;
      \For{each other child $c$ of $v$ from left to right}
         \State $T = \algmerge(T, \algtable(c))$;
      \EndFor
      \State \textbf{return} $\algadjust(T)$;
   \ElsIf{$v$ is a vertex representing a face $f$ and $f$ encloses a level $i+1$ component $C$}
      \State let $w$ be the root of the tree corresponding to $C$;
      \State \textbf{return} $\algadjust(\algcontract(\algtable(w)))$;
   \ElsIf{$v$ is a level 1 leaf}
      \State \textbf{return} a table for the edge $(x,y)$;
   \Else
      \State /* $v$ is a level $i>1$ leaf */
      \State let $f$ be the level $i$ face enclosing the component for $v$;
      \State let the labels of the children of $\vertex(f)$ be $(z_1,z_2),(z_2,z_3),\dots,(z_m,z_{m+1})$;
      \If{$y$ is adjacent to some $z_r$ for $\lbn(v) \leq r \leq \rbn(v)$}
         \State let $p$ be the least such $r$;
      \Else
         \State $p = \rbn(v)$;
      \EndIf
      \State /* note $z_p$ is a point between level $i-1$ vertices adjacent to $x$ and level $i-1$ vertices adjacent to $y$ in $\slice(v)$ */
      \State $T = \algcreate(v,p)$;
      /* extend the leftmost $p$ tables to include $x$ and edges from $x$ to $z_r$ for $r \leq p$, and merge with $T$ */
      \State let $j = p-1$;
      \While{$j \geq \lbn(v)$}
         \State $T = \algmerge(\algextend(x,\algtable(u_j)),T)$ where $u_j$ is the $j$th child of $\vertex(f)$;
         \State $j = j-1$;
      \EndWhile

      /* extend the remaining tables to include $y$ and edges from $y$ to $u_r$ for $r \geq p$, and merge with $T$ */
      \State $j = p$;
      \While{$j < \rbn(v)$}
         \State $T = \algmerge(T,\algextend(y,\algtable(u_j)))$ where $u_j$ is the $j$th child of $\vertex(f)$;
         \State $j = j + 1$;
      \EndWhile
      \State \textbf{return} $T$;
   \EndIf
\EndProcedure
\end{algorithmic}
\caption{The \textit{table} procedure.}
\label{alg:table}
\end{figure}

\begin{figure}[!htbp]
\begin{algorithmic}
\small
\Procedure{adjust}{$T$}
   \State let $L$ and $R$ be the left and right boundaries, respectively, for the slice that $T$ represents;
   \State let $x$ and $y$ be the highest level vertices in $L$ and $R$, respectively;
   \If{$x \neq y$ and $(x,y)$ is an edge in the graph}
      \State add 1 to any table entry in $T$ where $x$ and $y$ are both included;
   \EndIf
   \State \textbf{return} $T$;
\EndProcedure
\end{algorithmic}
\caption{The \textit{adjust} procedure.}
\label{alg:adjust}
\end{figure}

\begin{figure}[!htbp]
\begin{algorithmic}
\small
\Procedure{merge}{$T_1,T_2$}
   \State let $T$ be an initially empty table;
   \State let $L$ and $M$ be the left and right boundaries of the slice that $T_1$ represents;
   \State let $M$ and $R$ be the left and right boundaries of the slice that $T_2$ represents;
   \For{each subset $A$ of $L \cup R$}
      \For{each $k'=0, \dots, \max_{T_1}k' + \max_{T_2}k' - |M|$}
         \State let $V$ be an initially empty list;
         \For{each subset $B$ of $M$}
            \State let $n = 0$;
            \State let $x$ and $y$ be the top level vertices in $L$ and $R$, respectively;
            \If{$x = y$ and $x$ and $y$ are both in $A$}
               \State let $n = 1$;
            \EndIf
            \For{each $k_1,k_2$ satisfying $k_1 + k_2 - |B| - n = k'$}
               \State let $m$ be the number of edges between vertices in $B$;
               \State let $v = T_1((A \cap L) \cup B, k_1) + T_2((A \cap R) \cup B, k_2) - m$;
               \If{$v$ is not undefined}
                  \State add $v$ to $V$;
               \EndIf
            \EndFor
         \If{$V$ is not empty}
            \State let $T(A, k') = \max_vV$;
         \Else
            \State let $T(A, k')$ be undefined;
         \EndIf
         \EndFor
      \EndFor
   \EndFor
   \State \textbf{return} $T$;
\EndProcedure
\end{algorithmic}
\caption{The \textit{merge} procedure.}
\label{alg:merge}
\end{figure}

\begin{figure}[!htbp]
\begin{algorithmic}
\small
\Procedure{contract}{$T$}
   \State let $C$ be the level $i+1$ component such that $T$ represents the slice of $C$;
   \State let $f$ be the level $i$ face that contains $C$;
   \State let $L$ and $R$ be the left and right boundaries of the slice of $f$;
   \State let $(z,z)$ be the label of the tree root such that $T$ is a table for $\slice((z,z))$;
   \State let $T'$ be an initially empty table;
   \For{each subset $A$ of $L \cup R$}
      \For{each $k'=0,\dots,\max_Tk'$}
         \State let $T'(A, k') = \max\{T(A, k'), T(A \cup \{z\}, k')\}$;
      \EndFor
   \EndFor
   \State \textbf{return} $T'$;
\EndProcedure
\end{algorithmic}
\caption{The \textit{contract} procedure.}
\label{alg:contract}
\end{figure}

\begin{figure}[!htbp]
\begin{algorithmic}
\small
\Procedure{extend}{$z,T$}
   \State let $T'$ be a table that is initialized with every entry in $T$;
   \State let $L$ and $R$ be the boundaries for the slice represented by $T$;
   \For{each subset $A$ of $L \cup R$}
      \For{each $k'=0,\dots,\max_Tk'$}
         \If{$T(A,k'-1)$ is not undefined}
            \State let $m$ be the number of edges between $z$ and every vertex in $A$;
            \State let $T'(A \cup \{z\},k') = T(A,k'-1) + m$;
         \Else
            \State let $T'(A \cup \{z\},k')$ be undefined;
         \EndIf
      \EndFor
   \EndFor
   \State \textbf{return} $T'$;
\EndProcedure
\end{algorithmic}
\caption{The \textit{extend} procedure.}
\label{alg:extend}
\end{figure}

\end{document}